\documentclass[conference]{ieeeconf}
\IEEEoverridecommandlockouts

\usepackage{amsmath,amssymb,amsfonts}

\usepackage{graphicx}
\usepackage{textcomp}
\usepackage{xcolor}

\usepackage{algorithm}
\usepackage{algpseudocode}
\usepackage{empheq}
\usepackage{multicol}
\usepackage{cite}
\makeatletter
\let\NAT@parse\undefined
\makeatother
\usepackage{hyperref}

\usepackage{amsthm}

\def\BibTeX{{\rm B\kern-.05em{\sc i\kern-.025em b}\kern-.08em
    T\kern-.1667em\lower.7ex\hbox{E}\kern-.125emX}}

\begin{document}

\title{Formation Maneuver Control Based on the 
Augmented Laplacian Method
}
\author{Xinzhe Zhou$^{\dag}$, Xuyang Wang$^{\dag}$, Xiaoming Duan$^{\dag}$, Yuzhu Bai$^{\ddag}$, and Jianping He$^{\dag}$
\thanks{${\dag}$: the Department of Automation, Shanghai Jiao Tong University, Shanghai, China. E-mail: \{zhou\_xinzhe, 2674789652, xduan, jphe\}@sjtu.edu.cn}
\thanks{${\ddag}$: the College of Aerospace Science and Engineering, National University of Defense Technology, Changsha, China. E-mail: baiyuzhu06@nudt.edu.cn}
}
\maketitle

\newtheorem{theorem}{\textbf{Theorem}}
\newtheorem{lemma}{\textbf{Lemma}}
\theoremstyle{definition}
\newtheorem{definition}{\textbf{Definition}}
\newtheorem{assumption}{\textbf{Assumption}}
\newtheorem{remark}{\textbf{Remark}}

\begin{abstract}

This paper proposes a novel formation maneuver control method for both 2-D and 3-D space, which enables the formation to translate, scale, and rotate with arbitrary orientation. The core innovation is the novel design of weights in the proposed augmented Laplacian matrix. Instead of using scalars, we represent weights as matrices, which are designed based on a specified rotation axis and allow the formation to perform rotation in 3-D space. 
To further improve the flexibility and scalability of the formation, the rotational axis adjustment approach and dynamic agent reconfiguration method are developed, allowing formations to rotate around arbitrary axes in 3-D space and new agents to join the formation. 
Theoretical analysis is provided to show that the proposed approach preserves the original configuration of the formation. The proposed method maintains the advantages of the complex Laplacian-based method, including reduced neighbor requirements and no reliance on generic or convex nominal configurations, while achieving arbitrary orientation rotations via a more simplified implementation. Simulations in both 2-D and 3-D space validate the effectiveness of the proposed method.

\end{abstract}


\section{Introduction}
In recent years, formation control of multi-agent systems has gained significant attention due to its wide range of applications in various fields, such as drone swarms\cite{b1}, AUV formations \cite{b2}, robotic cooperation\cite{b3}, etc. While formation shape control remains essential for coordinated tasks, most scenarios increasingly require dynamic adaptability. This promotes the shift from formation shape control to formation maneuver control, which enables formations to perform continuous shape variation.
Formation maneuver control faces the new challenge of maintaining invariant geometric features during the maneuver process.

Although conventional consensus-based formation control methods are capable of tracking time-varying formations, they usually require agents to have explicit knowledge of target positions\cite{b4, b5}.  
Nevertheless, in most cases, agents are mainly restricted to gathering information from their neighboring agents rather than having access to global information \cite{b6}.
To deal with this issue, displacement-based\cite{b7,b8}, distance-based\cite{b9}, and bearing-based methods\cite{b10,b11} utilize the information from neighboring agents. They define target formations through invariant constraints on inter-agent displacements, distances, and bearings.
Ren et al. \cite{b7, b8} proposed displacement-based method which successfully tracks time-varying formation translation yet encounters difficulties when dealing with time-varying scaling and rotation.
De Marina et al. \cite{b9} put forward a distance-based method that is capable of handling time-varying formation translation and rotation while being ineffective in managing scaling.
Zhao et al. \cite{b10, b11} introduced bearing-based method which can accommodate time-varying formation translation and scaling while being unable to execute formation rotation.

Recent researches have improved formation maneuverability by introducing advanced constraints, such as similarity-based rigidity constraints, complex Laplacian, equilibrium stress matrices, and barycentric coordinates. 
The formation maneuver control method\cite{b12} based on similarity-based rigidity constraints effectively enables formation shape variation under translation, rotation, and scaling, but the inherent rigidity may pose challenges in terms of flexibility for reconfiguration.
The complex Laplacian approach\cite{b13,b14,b15,b16} enables formation shape variation under translation, rotation, and scaling. 
Because of the property of complex numbers, this method is restricted to 2-D space.
Xu et al. \cite{b17} extended the complex Laplacian-based method to 3-D space by introducing an additional dimension. To achieve rotation with any orientation in 3-D space, the method requires constructing three constant nominal configurations, which leads to high complexity and redundancy.
In contrast, formation maneuver control methods based on equilibrium stress matrices\cite{b18,b19,b20,b21} and barycentric coordinates\cite{b22,b23} support any-dimensional formation shape variation. However, these methods still have the following two issues. First, followers require at least $d+1$ neighbors in $d$-dimensional space. Second, nominal configurations must satisfy generic or rigid conditions to construct an invertible follower-follower matrix, which ensures unique and localizable formations.

In this paper, we present a novel formation maneuver control method based on the \textit{augmented Laplacian matrix}. The proposed approach enables formations to simultaneously execute translational, rotational, and scaling maneuvers in both 2-D and 3-D space. Compared with existing methods, our main contributions are as follows:
\begin{itemize}
\item  Instead of scalar values, the weights of inter-agent constraints are represented as matrices, which enables the formation to rotate about the given rotation axis.
\item To achieve arbitrary orientation rotations of the formation, we propose a dynamic rotation axis adjustment approach. We also develop a dynamic agent reconfiguration method that allows seamless integration of new agents into the formation. It is proved that both methods preserve the invariance of formation's configuration.
\item We establish a theoretical connection between the proposed method and existing complex Laplacian method, thereby showing that the complex Laplacian method is actually a special case of our method.
\end{itemize}




The rest of this paper is organized as follows. Section \hyperref[Section II]{II} presents some preliminaries and the objectives of this paper. Section \hyperref[Section III]{III} proposes the augmented Laplacian matrix, and control protocols for both leaders and followers. Further, we give methods for adjusting the rotation axis of formation and dynamic agent reconfiguration. Section \hyperref[Section IV]{IV} reveals the relationship between the proposed method and the 2-D complex Laplacian method. Simulation results are presented in Section~\hyperref[Section V]{V} to validate the effectiveness of our proposed method. Lastly, the conclusion is given in Section \hyperref[Section VI]{VI}.

\section{Preliminaries and Problem Formulation}
\label{Section II}
\subsection{Notations}
Let $\mathbb{R}$ be the set of real numbers, and $\mathbb{C}$ be the set of complex numbers. $A^{\top}$ is the transpose of the real matrix $ A \in \mathbb{R}^{m \times n}$.
Consider a multi-agent system with $n$ agents in $\mathbb{R}^d$, where $d \in \{2,3\}$ and $n \geq 3$. Let $p_i \in \mathbb{R}^d$ be the position of the $i$-th agent, $r_i$ be the nominal position and $p^*_i(t)$ be the target position at time $t$. Denote the interaction graph among the agents as $\mathcal{G} = (\mathcal{V}, \mathcal{E})$, where $ \mathcal{V} = \{1, \ldots, n\} $ and $\mathcal{E} \subseteq \mathcal{V} \times \mathcal{V} $ are the sets of agents and edges, respectively. This paper only considers \textit{2-rooted} graphs, i.e., there is a subset of two agents, and every other agent is 2-reachable from this subset \cite{b13}.

Denote $(\mathcal{G}, p)$ as a formation of $n$ agents, where $\mathcal{G}$ is the interaction graph and $p = [p_1^\top, \ldots, p_n^\top]^\top \in \mathbb{R}^{nd}$ is the configuration of the formation.
Similarly, let $(\mathcal{G}, p^*(t))$ and $(\mathcal{G}, r)$ be the target formation and nominal formation, where $p^*(t) = [p_1^*(t)^\top, \ldots, p_n^*(t)^\top]^\top$ and $r = [r_1^\top, \ldots, r_n^\top]^\top$.
A nominal configuration $r$ is called generic if the coordinates $r_1,\ldots,r_n$ do not satisfy any nontrivial algebraic equation with integer coefficients \cite{b17}. The centroid $r_c \in \mathbb{R}^d$ of the nominal formation $(\mathcal{G}, r)$ is defined as
\begin{equation}
    \label{center of r}
    r_c = \frac{1}{n} \sum_{i \in \mathcal{V}} r_i.
\end{equation}

Let $ \mathcal{V}_f = \{1, \ldots, n_f\} $ and $\mathcal{V}_l = \mathcal{V} \setminus \mathcal{V}_f$ be the sets of followers and leaders.
Let $p_f = \begin{bmatrix} p_1^\top, \ldots, p_{n_f}^\top \end{bmatrix}^{\top} $ and $p_l = \begin{bmatrix} p_{n_f + 1}^{\top}, \ldots, p_n^{\top} \end{bmatrix}^{\top} $ be the configurations of followers and leaders, respectively.
Similarly, let $p_f^*(t)$, $p_l^*(t)$ and $r_f$, $r_l$ be the target configuration and nominal configuration.

 Let $\zeta = [\zeta_x, \zeta_y, \zeta_z]^\top \in \mathbb{R}^3$ be the axis of rotation in 3-D space, and $\zeta^{\times}$ be the skew-symmetric matrix of vector $\zeta$. Let $\mathcal{N}_i \triangleq \{ j \in \mathcal{V} : (i, j) \in \mathcal{E} \}$ be the set of neighbors of agent $i$. Let $I_d \in \mathbb{R}^{d \times d}$ be an identity matrix of dimension $d \times d$, and $\mathbf{1}_d$ be the all-one vector of dimension $d$. Let $\otimes$ be the Kronecker product. 

\subsection{Invariant Constraints}
A main challenge in formation maneuver control lies in the design of constraints, which maintain invariance under formation transformations. Specifically, the constraints are realized through the design of edge weights $w_{ij}$ in the formation's interaction graph.
For instance, the constraint proposed in \cite{b11} maintains translation and scaling invariance through bearing-based weights $w_{ij}$, while the constraint in \cite{b18} extends the invariance to affine transformations. 

A general form of constraints is formulated as
\begin{equation}
    \label{general constraints}
    \sum_{j\in \mathcal{N}_i} w_{ij}(p_j - p_i) = 0,
\end{equation}
where $w_{ij}$ are the weights on the edges $(i, j)$. 
Then, the \textit{constraint matrix} $W$ is defined as
\begin{equation}
W(i, j) = 
\begin{cases} 
w_{ij} & \text{if } i \neq j \text{ and } j \in \mathcal{N}_i \\\
0 & \text{if } i \neq j \text{ and } j \notin \mathcal{N}_i \\\
-\sum_{k \in \mathcal{N}_i} w_{ik} & \text{if } i = j
\end{cases}
\label{laplacian matrix}
\end{equation}
The shape of the matrix \(W\) differs among different methods. For instance, in the complex Laplacian methods \cite{b13,b14,b15,b16}, \(W \in \mathbb{C}^{n\times n}\), while in the methods based on equilibrium stress matrices \cite{b18,b19,b20,b21}, \(W \in \mathbb{R}^{nd\times nd}\). In general, \eqref{general constraints} can be written in the matrix form as
\begin{equation}
    \label{general constraints in matrix}
    Wp = 0.
\end{equation}


Then, given a configuration $p$, an immediate question is how to generate weights $w_{ij}$ under constraints \eqref{general constraints}. Specifically, the work \cite{b13} addressed this by decomposing the problem into cases where agents have exactly two neighbors and those with more than two neighbors. First, for agent $i$ with two neighbors, $j, k$, the weights are designed as
\begin{equation}
    \label{two neighbor construct}
    w_{ij}(p_j - p_i)+w_{ik}(p_k - p_i)=0.
\end{equation}
Second, for agent $i$ with more than two neighbors, $i_1, \ldots, i_m$ where $m>2$, the weights are constructed via 
\begin{subequations}
\renewcommand{\theequation}{\theparentequation-\alph{equation}}
\label{more than two neighbor construct}
\begin{empheq}[left=\empheqlbrace]{align}
   &\xi_h = [0, 0, \ldots, w_{ii_j}^{(h)}, 0, \ldots, w_{ii_k}^{(h)}, \ldots, 0], \label{(9)}\\
&[w_{ii_1}, w_{ii_2}, \ldots, w_{ii_m}]=\sum_{h = 1}^{C_{m}^2}\xi_h,
\end{empheq}
\end{subequations}
which, in essence, converts the problem into combinations of two-neighbor cases.

The constraint matrix, \(W\), is partitioned into follower-leader blocks as
\begin{equation}
    \label{(5)}
    W = 
    \begin{bmatrix}
    W_{ff} & W_{fl} \\
    W_{lf} & W_{ll}
    \end{bmatrix},
\end{equation}
where \( W_{ff} \in \mathbb{R}^{n_fd \times n_fd} \) and \( W_{ll} \in \mathbb{R}^{n_ld \times n_ld} \) describe the follower-follower and leader-leader blocks, respectively. The cross blocks $W_{fl} \in \mathbb{R}^{n_fd \times n_ld}$ and $W_{lf} \in \mathbb{R}^{n_ld \times n_fd}$ describe leader-follower interactions.

\begin{definition}
    A nominal formation $(\mathcal{G}, r)$ is termed \textit{localizable} if the follower-follower matrix $W_{ff}$ is non-singular. 
\end{definition}

When the formation configuration $p$ is at nominal configuration $r$, \eqref{general constraints in matrix} decomposes into follower-leader blocks as
\begin{equation}
    \begin{bmatrix}
    W_{ff} & W_{fl} \\
    W_{lf} & W_{ll}
    \end{bmatrix}
    \begin{bmatrix}
    r_f \\
    r_l
     \end{bmatrix}
    =0,
\end{equation}  
which further implies $W_{ff}r_f + W_{fl}r_l = 0$. Given the leaders' nominal configuration $r_l$, non-singularity of the matrix $W_{ff}$ guarantees the existence of a unique solution for the followers' configuration $r_f$. Specifically, 
\begin{equation}
    \label{uniqueness of follower nominal configuration}
     r_f = -W_{ff}^{-1}W_{fl}r_l.
\end{equation}

\begin{assumption}
    \label{assumption 1}
     The nominal formation $(\mathcal{G}, r)$ is localizable, i.e., the matrix $W_{ff}$ is non-singular. 
\end{assumption}
In Sec. \hyperref[Section III]{III}, we will show the rationality of this assumption.

\subsection{Target Formation and Objectives}
The translation and scaling parameters are denoted by $T(t) \in \mathbb{R}^d$ and $k(t) \in \mathbb{R}$. In addition, denote $R_\zeta(t) \in \mathbb{R}^{d \times d}$ as the rotation matrix about axis $\zeta$. Let $\theta$ be the rotation angle.
By Rodrigues' rotation formula, the rotation matrix $R_\zeta$ is expressed as 
\begin{equation}
    \label{rotation matrix}
    R_\zeta = I_d+(\sin\theta)\zeta^{\times}+(1 - \cos\theta)(\zeta^{\times})^2.
\end{equation}

Subsequently, according to \cite{b17}, the time-varying target configuration $p^*(t)$ is defined similarly based on the nominal configuration $r$ as
\begin{equation}
    \label{all target configuration}
    p^*(t) = \underbrace{\mathbf{1}_n \otimes (r_c + T(t))}_{\alpha^*(t)} + \underbrace{k(t)( I_n \otimes R_\zeta(t))(r - \mathbf{1}_n \otimes r_c)}_{\beta^*(t)},
\end{equation}
where $\alpha^*(t)$ is the translation term and $\beta^*(t)$ is the scaling and rotation term.

\begin{definition}
    \label{def of similar formation}
    For a nominal formation $(\mathcal{G},r)$, a target formation $(\mathcal{G},p^*)$ is termed \textit{similar} to $(\mathcal{G},r)$, if there is a specific rotation axis \( \zeta \), such that the constraint matrix \(W\) generated by the nominal formation $(\mathcal{G},r)$ satisfies
    \begin{equation}
        \label{def of similar}
        Wp^* = 0.
    \end{equation}
\end{definition}
During formation maneuvers, the time-varying target formation $(\mathcal{G}, p^*(t))$ should remain \textit{similar} to $(\mathcal{G},r)$ for all time, i.e.,  $p^*(t) \in \ker(W) \text{ for all } t$.

\begin{assumption}
    \label{assumption 2}
    Only leaders get access to the target positions, which means only leaders know the time-varying parameters $T(t), k(t), R_\zeta(t)$ in \eqref{all target configuration}. Followers only know the relative positions and velocities of their neighbors. 
\end{assumption}

Under \hyperref[assumption 1]{Assumption 1} and \hyperref[assumption 2]{Assumption 2}, and leveraging \eqref{all target configuration} and \eqref{def of similar}, the target configurations of the leader and follower groups are written as
\[
\begin{cases}
    p_l^*=&\mathbf{1}_{n_l}\otimes(r_c + T(t))\\&+k(t)(I_{n_l}\otimes R_\zeta(t))(r_l - \mathbf{1}_{n_l}\otimes r_c),\\
    p_f^*=&-W_{ff}^{-1}W_{fl}p_l^*.
\end{cases}
\]

Then, the objectives of this paper are as follows.
\begin{itemize}
    \item Design a constraint matrix $W$ to ensure that the target formation configuration $p^*(t)$ satisfies \eqref{def of similar} for all time during formation maneuvers.
    
    \item Design control protocols for the leaders and followers such that
    \begin{subequations}
    \renewcommand{\theequation}{\theparentequation-\alph{equation}}
    \begin{empheq}[left=\empheqlbrace]{align}
    &\lim_{t \to \infty} (p_l(t) - p_l^*(t)) = 0, \label{control objective for leaders}\\
    &\lim_{t \to \infty} (p_f(t) + W_{ff}^{-1} W_{fl}p_l(t)) = 0. \label{control objective for followers}
    \end{empheq}
    \end{subequations}
    
    \item Propose a method that enables the formation to adjust its rotation axis at any time to rotate with any orientations in 3-D space.
    
    \item Develop a method enabling the real-time and seamless integration of new agents into the formation.
    
\end{itemize}

\section{Methodology}
\label{Section III}
\subsection{Augmented Laplacian Matrix}
In 3-D formation maneuver control, a fundamental challenge is achieving arbitrary formation rotations. Existing methods often lack systematic and straightforward solutions for this problem. To address this gap, we propose the \textit{augmented Laplacian matrix}, where weights $w_{ij}$ are matrices in $\mathbb{R}^{d\times d}$ rather than scalars. 
\begin{definition}
    Given rotation axis $\zeta$, the \textit{augmented Laplacian matrix} $W$ is expressed in the form of \eqref{laplacian matrix}, and each weight $w_{ij}$ is defined by
    \begin{equation}
        \label{form of wij}
        w_{ij} = a_{ij} I_d + b_{ij} \zeta\zeta^\top + c_{ij} \zeta^{\times},
    \end{equation}
     where coefficients \( a_{ij}, b_{ij}, c_{ij} \in \mathbb{R} \).
\end{definition}

As shown in \eqref{two neighbor construct} and \eqref{more than two neighbor construct}, regardless of the number of neighbors agent \(i\) has, the problem can be split into combinations of two-neighbor cases. Therefore, assuming that agents \(j, k\) are the neighbors of agent \(i\), by \eqref{general constraints}, \eqref{two neighbor construct} and \eqref{more than two neighbor construct}, we have
\begin{align}
\label{existence}
    &(a_{ij} I_d + b_{ij} \zeta\zeta^\top + c_{ij} \zeta^{\times})(p_j - p_i) \nonumber\\ &+ (a_{ik} I_d + b_{ik} \zeta\zeta^\top + c_{ik} \zeta^{\times})(p_k - p_i) = 0.
\end{align}
There are six unknowns (\(a_{ij}, b_{ij}, c_{ij}, a_{ik}, b_{ik}, c_{ik}\)) and three equations in~\eqref{existence}. By the rank-nullity theorem, a non-trivial solution exists if and only if the coefficient matrix has rank less than three. Under \hyperref[assumption 1]{Assumption 1}, feasible weight matrices \(w_{ij}\) can always be constructed to satisfy \eqref{general constraints}.

\begin{theorem}
    \label{commutative law}
    Given rotation axis $\zeta$, the weight $w$ defined in~\eqref{form of wij} commutes with the rotation matrix $R_\zeta$, i.e., 
    \[
        wR_\zeta = R_\zeta w.
    \]
\end{theorem}
\begin{proof}
    According to \eqref{rotation matrix} and \eqref{form of wij}, we have
    \begin{align*}
        wR_\zeta - R_\zeta w =& b\sin\theta (\zeta\zeta^{\top} \zeta^{\times}-\zeta^{\times}\zeta\zeta^{\top}) + \\ &a(1-\cos\theta)(\zeta\zeta^{\top}(\zeta^{\times})^2-(\zeta^{\times})^2\zeta\zeta^{\top}),
    \end{align*}
    Utilizing the property of skew-symmetric matrices, we have 
    \[\zeta\zeta^{\top}\zeta^{\times} = \zeta^{\times}\zeta\zeta^{\top} = 0, \quad \zeta\zeta^{\top}(\zeta^{\times})^2 = (\zeta^{\times})^2\zeta\zeta^{\top} = 0.\]
    Thus, we have 
    \[wR_\zeta - R_\zeta w = 0 \implies wR_\zeta = R_\zeta w.\]
\end{proof}

Using \hyperref[commutative law]{Theorem 1}, the weight design \eqref{form of wij} is further proved to satisfy~\eqref{def of similar}.

\begin{theorem}
    \label{theorem for similar}
    The time-varying target formation $(\mathcal{G}, p^*(t))$ in \eqref{all target configuration} is similar to its nominal formation $(\mathcal{G}, r)$ with weights $w_{ij}$ in the form of \eqref{form of wij}.
\end{theorem}

\begin{proof}
    For the nominal formation $(\mathcal{G}, r)$, the constraints \eqref{general constraints} are expressed as
    \begin{equation}
    \label{general constraints for nominal constraints}
        \sum_{j\in \mathcal{N}_i} w_{ij}(r_j - r_i) = \sum_{j} w_{ij} r_j = 0.
    \end{equation}
   
    First, for the translation term in \eqref{all target configuration}, we have 
    \[
    W \alpha^*(t) = \begin{bmatrix} \sum_{j} w_{1j} (r_c + T(t)) \\ \vdots \\ \sum_{j} w_{nj} (r_c + T(t)) \end{bmatrix} = 0.
    \]   
    Second, for the scaling and rotation term , we have
    \begin{align*}
        W \beta^*(t) = 
        \begin{bmatrix} 
        \sum_{j} w_{1j} k(t)R_\zeta(t)(r_j - r_c) \\ \vdots \\ 
        \sum_{j} w_{nj} k(t)R_\zeta(t)(r_j - r_c) 
        \end{bmatrix}.
    \end{align*}
    For each row, by \hyperref[commutative law]{Theorem 1}, $w$ commutes with $R_\zeta$  
    \begin{align*}
        &\sum_{j} w_{ij} k(t)R_\zeta(t)(r_j - r_c)\\ 
        &= k(t)R_\zeta(t) \sum_{j} w_{ij}  (r_j - r_c) = 0.
    \end{align*}
    Thus, we obtain
    \begin{align*}
        W [\alpha^*(t) + \beta^*(t)] = W p^*(t) = 0.
    \end{align*}
   
Therefore, given the augmented Laplacian matrix $W$ with entries in the form of \eqref{form of wij}, the target formation $(\mathcal{G}, p^*(t))$ is similar to the nominal formation $(\mathcal{G}, r)$ for all time.
    
\end{proof}

\subsection{Control Protocols}

Consider that each agent is modeled as a single-integrator
\begin{equation}
    \label{agent model}
    \dot{p}_i(t) = v_i(t),
\end{equation}
where $v_i(t)$ is the control protocol to be designed. Subsequently, we will present the specific forms for leaders and followers respectively.

\textbf{Control Protocol for the Leaders}. According to \cite{b17}, a feasible control protocol for the leaders is obtained as
\begin{equation}
    \label{control protocol for leaders}
    v_i = \begin{bmatrix}
        -tanh(x_i - x_i^*) + \Dot{x}_i^* \\
        -tanh(y_i - y_i^*) + \Dot{y}_i^* \\
        -tanh(z_i - z_i^*) + \Dot{z}_i^*
    \end{bmatrix},
    i \in \mathcal{V}_l,
\end{equation}
where \( v_i \in \mathbb{R}^d \) is the control input, \(\tanh(\cdot)\) is the hyperbolic tangent function, and \(\dot{x}_i^*\), \(\dot{y}_i^*\), \(\dot{z}_i^*\) are the target velocities.

\begin{lemma}[Lemma 3 \cite{b17}]
    The leaders group achieves their control objective \eqref{control objective for leaders} under control protocol \eqref{control protocol for leaders}.
\end{lemma}

\textbf{Control Protocol for the Followers}. Under \hyperref[assumption 1]{Assumption~1}, the followers can utilize constant weights in $W_{ff}$ and $W_{fl}$ to achieve formation maneuver control without knowing the time-varying parameters $T(t), k(t), R_\zeta(t)$ in \eqref{all target configuration}.
The distributed control protocol for followers is given by
\begin{equation}
    \label{control protocol for followers}
    \Dot{p}_i = \gamma_i^{-1}\sum_{j\in\mathcal{N}_i}w_{ij}[\alpha(p_i - p_j) - \Dot{p}_j], i \in \mathcal{V}_f,
\end{equation}
where ${\gamma_i} = \sum_{j\in\mathcal{N}_i}w_{ij}$ represents the sum of weights for agent $i$'s neighbors, and $\alpha \in \mathbb{R^{+}}$ is a positive control gain parameter. In matrix form, \eqref{control protocol for followers} is expressed as
\begin{equation}
    \label{control protocol for followers in matrix}
    W_{ff}\Dot{p}_f + W_{fl}\Dot{p}_l = -\alpha(W_{ff}p_f + W_{fl}p_l).
\end{equation}

\begin{lemma}
    The followers group achieves their control objective \eqref{control objective for followers} under control protocol \eqref{control protocol for followers in matrix}.
\end{lemma}

\begin{proof}
    Let \( e_f \) represent the tracking error for followers, which is defined by
    \begin{equation*}
        e_f = p_f + W_{ff}^{-1}W_{fl}p_l.
    \end{equation*}
    Differentiating this equation gives
    \begin{equation*}
        \Dot{e}_f = \Dot{p}_f + W_{ff}^{-1}W_{fl}\Dot{p}_l = -\alpha(W_{ff}p_f + W_{fl}p_l) = -\alpha e_f.
    \end{equation*}
    Consider a Lyapunov function
    \begin{equation*}
        V = \frac{1}{2}e_f^{\top}e_f.
    \end{equation*}
    The time derivative of $V$ is
    \begin{equation*}
        \Dot{V} = e_f^{\top}\Dot{e}_f = -\alpha e_f^{\top}e_f.
    \end{equation*}
    Since $\Dot{V} < 0$ for $e_f \neq 0$, the desired conclusion holds.

\end{proof}


\subsection{Rotation Axis Adjustment}
Since the augmented Laplacian matrix is related to $\zeta$, the proposed augmented Laplacian method confines the formation to rotate around a given axis $\zeta$. Thus, a natural question is how to enable the formation to rotate with any orientations in 3-D space. To address this issue, we introduce a method for adjusting the rotation axis during the formation maneuver, while maintaining the invariance of the formation geometry.

First, at time \( t_1 \), the formation reaches the target configuration, i.e., \( p(t_1) = p^*(t_1) \). At this specific time \( t_1 \), the rotation axis is adjusted from \( \zeta \) to \( \zeta_{u} \). Next, based on the current formation \( (\mathcal{G}, p(t_1)) \) and rotation axis \( \zeta_{u} \), the augmented Laplacian matrix \( W_{u} \) must be reconstructed. Finally, in accordance with \eqref{all target configuration}, the target configuration \( p^*_{u}(t) \) for the subsequent formation maneuver is expressed as
\begin{equation}
    \label{target configuration after change rotation axis}
    p^*_{u}(t) = \alpha^*_{u}(t) + \beta^*_{u}(t),
\end{equation}
where $\alpha^*_{u}(t) = \mathbf{1}_n \otimes (r_{c_{u}} + T_{u}(t))$ and $\beta^*_{u}(t) = k_{u}(t)(I_n \otimes R_{\zeta_{u}}(t))(r_{u} - \mathbf{1}_n \otimes r_{c_{u}})$.

\begin{theorem}
    After adjusting the rotation axis $\zeta$ to $\zeta_{u}$ , the new target formation $(\mathcal{G}, p^*_{u}(t))$ is still similar to the nominal formation $(\mathcal{G}, r)$.
\end{theorem}

\begin{proof}
    From \eqref{target configuration after change rotation axis}, we have 
    \begin{align*}
        &r_{u} = p^*(t_1) = \alpha^*(t_1) + \beta^*(t_1), \\
        &r_{c_{u}} = \frac{1}{n} \sum_{i \in \mathcal{V}} r_{{u}_i} = r_c + T(t_1) .
    \end{align*}
    Then, \eqref{target configuration after change rotation axis} is transformed to
    \begin{align*}
        p^*_{u}(t) &= \mathbf{1}_n \otimes (r_c + T(t_1) + T_{u}(t)) +  \\
        &\ k_{u}(t)k(t_1)(I_n \otimes R_{\zeta_{u}}(t)R_\zeta(t_1))(r - \mathbf{1}_n \otimes r_c),
    \end{align*}
    where $T(t_1) + T_{u}(t)$ is treated as $T_{u'}(t)$, $k_{u}(t)k(t_1)$ is treated as $k_{u'}(t)$ and $R_{\zeta_{u}}(t)R_\zeta(t_1)$ is treated as $R_{\zeta_{u'}}$. According to \hyperref[def of similar formation]{Definition 2}, the new target formation $(\mathcal{G}, p^*_{u}(t))$ is still similar to the nominal formation $(\mathcal{G}, r)$.
\end{proof}

\subsection{Dynamic Agent Reconfiguration}
In practical cases, the integration of new agents into a formation
is widely adopted, highlighting the significance of a scalable formation maneuver control method. In this subsection, we propose a method that allows new agents to seamlessly join the formation at any time for dynamic reconfiguration.

Specifically, the integration of a new agent into the formation involves the following steps. Initially, during the original formation maneuver, the new agent identifies its time-varying joining position. Upon determining this position, the agent utilizes the control law \eqref{control protocol for leaders} to reach its target position. Once the new agent arrives at the designated position, the augmented Laplacian matrix $W$ is reconstructed to accommodate the new agent while preserving the configuration of the original part. 
Ultimately, the integrated agent employs the distributed follower control protocol to maintain its relative position within the desired formation during subsequent maneuvers.

\begin{theorem}
After a new agent joins the formation, the method proposed above maintains the configuration of the original part, i.e., $p'_o = p$,
where $p'_o$ is the configuration after the addition of a new agent . 
\end{theorem}

\begin{proof}
After the integration of the new agent, the positions of all agents in the formation are represented by
\begin{equation}
p' = \begin{bmatrix}
p_{add} \\
p'_o
\end{bmatrix},
\end{equation}
where $p_{add}$ is the position of the added agent.

Then, the reconstructed augmented Laplacian matrix is presented in $W' = W_{ex} + \Delta W$, where $W'$ satisfies $W'p' = 0$, $W_{ex}$ represents the extended matrix of the original augmented Laplacian matrix $W$, with the specific form as
\begin{equation*}
W_{ex} = \begin{bmatrix}
0_{d \times d} & 0_{d \times nd} \\
0_{nd \times d} & W_{nd \times nd}
\end{bmatrix},
\end{equation*}
and $\Delta W$ represents the augmented Laplacian matrix exclusively concerning the added agent and its neighbors.
By constraints \eqref{general constraints}, considering the original configuration and the new agent, we have
\begin{equation}
    Wp = 0, W'p' = 0.
\end{equation}
Hence, $\Delta W$ must satisfy $\Delta Wp' = 0$. It follows that $W_{ex}p' = 0$, which further implies that $Wp'_o = 0$. Thus, we have
$$
W_{ff}p'_{of} + W_{fl}p'_{ol} = 0.
$$ 
Since the leaders' positions remain invariant during the integration of the new agent, i.e., $p'_{ol} = p_l$, it is concluded that the positions of the original followers in the formation also remain unchanged:
\begin{equation*}
p'_{of} = W_{ff}^{-1}W_{fl}p'_{ol} = W_{ff}^{-1}W_{fl}p_l = p_f.
\end{equation*}
Therefore, the integration of new agents maintains the invariance of the original formation configuration.
\end{proof}

\section{Relationships with 2-D Complex Laplacian Method}
\label{Section IV}

In this section, we show that the 3-D formation maneuver control method based on the augmented Laplacian matrix proposed in this paper can be treated as an extension of the 2-D complex Laplacian method \cite{b11,b15}.

In the complex Laplacian approach, the complex positions $p$ and weights $w_{ij}$ are expressed as $ p = x+y \iota $ and $w_{ij} = a_{ij} + b_{ij}\iota$, respectively, where \( x, y, a, b \in \mathbb{R} \) and \( \iota \) is the imaginary unit with \( \iota^2 = -1 \). For any position $p \in \mathbb{C}$, the weights $w_{ij}$ acting on $p$ are expressed as
\begin{equation}
    \label{wij p in complex}
    w_{ij}p = (a_{ij}x-b_{ij}y) +  (b_{ij}x + a_{ij}y)\iota.
\end{equation}
 
In 2-D plane rotations, all transformations occur around the Z-axis, which allows the 2-D complex Laplacian method to be considered as a special case of the proposed augmented Laplacian approach under rotations about the axis $\zeta = [0,0,1]^\top$. To align with the complex Laplacian method, we adjust the order of coefficients in equation \eqref{form of wij}
\begin{align}
    \label{form of wij in 2d space}
    w_{ij} &=  a_{ij} I_d + c_{ij} \zeta\zeta^\top + b_{ij} \zeta^{\times} \nonumber\\
           &=    \begin{bmatrix}
                a_{ij} & -b_{ij} & 0 \\
                b_{ij} & a_{ij} & 0 \\
                0 & 0 & a_{ij}+c_{ij}
                \end{bmatrix}.
\end{align}

Weights $w_{ij}$ constructed in \eqref{form of wij in 2d space} have the same property as complex weights. First, for any position $p = [x,y,0]^\top$ in 2-D space, the action of weights $w_{ij}$ on $p$ is
\begin{equation}
    \label{wij p in matrix}
    w_{ij} p = [a_{ij}x-b_{ij}y, b_{ij}x + a_{ij}y, 0]^\top,
\end{equation}
which is equivalent to the action \eqref{wij p in complex} in complex domain.
Additionally, by \hyperref[commutative law]{Theorem 1}, weights constructed in the proposed augmented Laplacian method also commute with rotation matrices $R_\zeta$, which is trivial for the complex Laplacian method due to the property of scalars. 

Therefore, our proposed method is also available in 2-D space. Let $R_z(t)$ be the rotation matrix around the Z-axis. Then the target configuration  $p^*$ is expressed as follows,
\begin{equation}
    p^*(t) = \mathbf{1}_n \otimes (r_c + T(t)) + k(t)(I_n \otimes R_z(t))(r - \mathbf{1}_n \otimes r_c).
\end{equation}
By \hyperref[theorem for similar]{Theorem 2}, the target formation $(\mathcal{G}, p^*)$ is similar to its nominal formation, i.e., $Wp^* = 0$. 
Therefore, the complex Laplacian method is equivalent to the method proposed in this paper in 2-D space, both in terms of the construction of weights and the representation of the target formation.

Notably, the key difference between the 2-D and 3-D scenarios lies in the form of the weights $w_{ij}$. For 2-D cases, the construction of weights $w_{ij}$ is simplified to the expression in \eqref{form of wij in 2d space}. In spite of this, the theorems and control protocols proposed in Sec.~\hyperref[Section III]{III} can be directly applied to 2-D cases without any modification.

\section{Simulation}
\label{Section V}
We will present two simulation examples in both 2-D and 3-D space. The multi-agent systems in 2-D and 3-D space both consist of three followers $\mathcal{V}_f=\{1,2,3\}$ and two leaders $\mathcal{V}_l=\{4,5\}$ shown in Fig.~\ref{fig:image3} and Fig.~\ref{fig:image1}, where the trajectories of the followers and leaders are marked in blue and red, and the trajectory of the new agent in 3-D simulation is marked in green. In Fig.~\ref{fig:image3} and Fig.~\ref{fig:image1}, the formations connected by gray lines are non-target formations, while the formations connected by green lines are target formations. The red arrows in Fig.~\ref{fig:image1} represent the rotation axis.

\subsection{Analysis of 2-D Simulation Results}
In the 2-D simulation, the nominal configuration $r = [r_1^\top,...,r_5^\top]^\top$ is designed as
\begin{align*}
    &r_1 = [0.5, 0.5, 0]^\top, \ r_2 = [0.5, -0.5, 0]^\top, r_3 = [0, 0, 0]^\top,\\  &r_4 = [-1, 1, 0]^\top, r_5 = [-1, -1, 0]^\top.
\end{align*}
with the positions of agents represented in 3-D to directly use the 3-D method proposed in this paper.

\begin{figure}[htbp]
    \centering
    \includegraphics[width=0.48\textwidth]{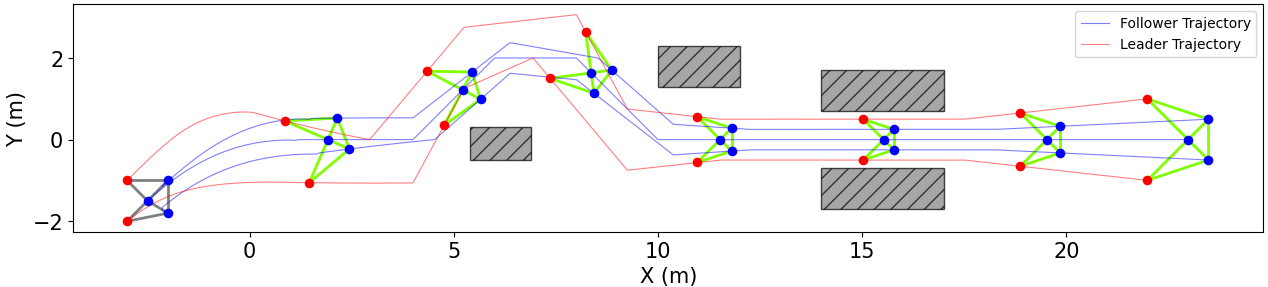}
    \caption{Trajectories in 2-D space (The dark areas stand for obstacles.)}
    \label{fig:image3}
\end{figure}

In this simulation, we verify the feasibility of our method in 2-D space. As shown in Fig.~\ref{fig:image3}, the 2-D formation can freely translate, scale and rotate to avoid obstacles. As shown in Fig.~\ref{fig:image4}, the control protocols we designed ensure that the formation quickly converges to the target formation.

\begin{figure}[htbp]
    \centering
    \includegraphics[width=0.48\textwidth]{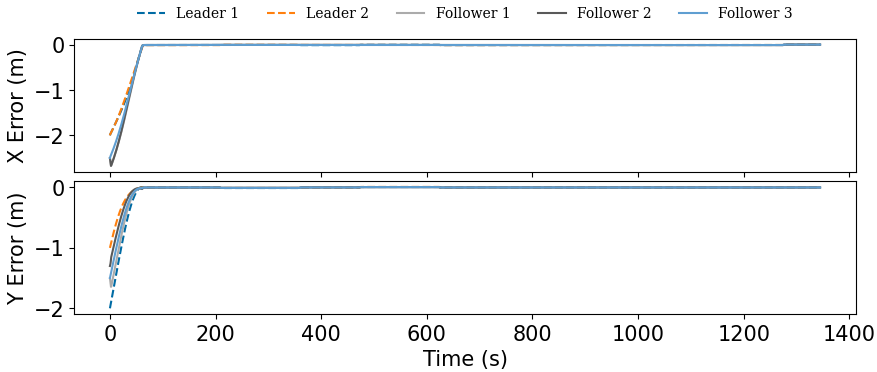}
    \caption{Tracking errors in 2-D space.}
    \label{fig:image4}
\end{figure}

\subsection{Analysis of 3-D Simulation Results}
In the 3-D simulation, the nominal configuration $r = [r_1^\top,...,r_5^\top]^\top$ is designed as
\begin{align*}
    &r_1 = [0.05, 0, 1]^\top, \ r_2 = [-0.05, 0, -1]^\top, r_3 = [1, \sqrt{3}, 0.05]^\top,\\ & r_4 = [1, -\sqrt{3}, -0.05]^\top, r_5 = [-2, 0, 0]^\top.
\end{align*}
Note that since there are only two leaders, the line connecting them must not be parallel to the axis of rotation.

\begin{figure}[htbp]
    \centering
    \includegraphics[width=0.45\textwidth]{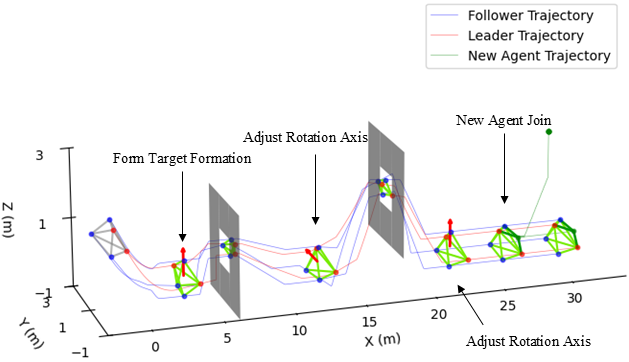}
    \caption{Trajectories in 3-D space (The dark areas stand for obstacles.)}
    \label{fig:image1}
\end{figure}

In this simulation, as shown in Fig.~\ref{fig:image1}, we not only verify that, based on the proposed control method, the formation is capable of translating, scaling and rotating around any axis in 3-D space, but also show that it is possible for the formation to realize reconfiguration.


Fig.~\ref{fig:image2} presents that tracking errors in each axis for both the agents in the formation and the new agent converge to zero quickly, which shows the effectiveness of our proposed control protocols.

\begin{figure}[htbp]
    \centering
    \includegraphics[width=0.48\textwidth]{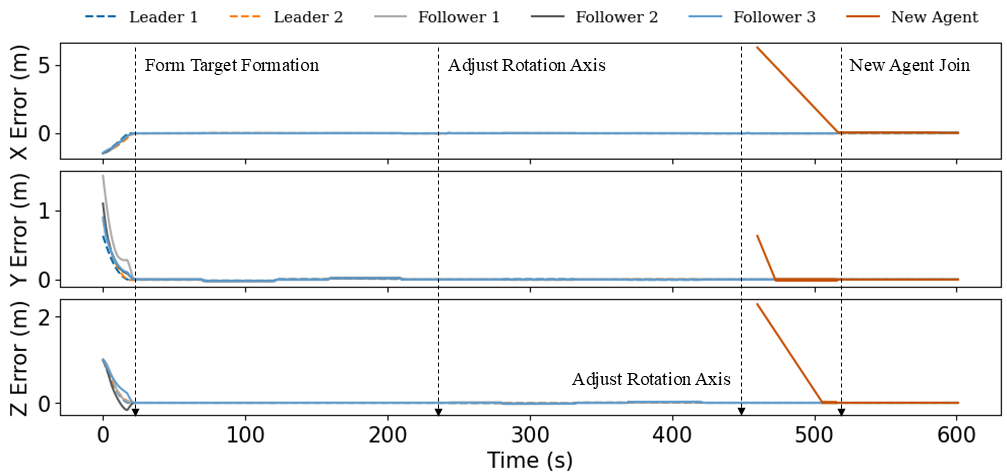}
    \caption{Tracking errors in 3-D space.}
    \label{fig:image2}
\end{figure}

Overall, the simulations demonstrate the effectiveness of our proposed formation maneuver control method in both 2-D and 3-D space. A key advantage of our method is the capacity to perform arbitrary rotations, which allows for complex maneuvers. Furthermore, our method achieves formation maneuverability with only two leaders in 3-D space. Additionally, the proposed method does not require convex or generic configurations, providing greater flexibility in formation design. These features collectively contribute to the robustness and versatility of our formation maneuver control method.

\section{Conclusions}
\label{Section VI}
This work introduces a general formation maneuver control method in both 2D and 3D space based on a novel constraint called augmented Laplacian matrix. The proposed method enables the formation to translate, scale, and rotate with any orientations. Moreover, for practical applications, we propose the dynamic agent reconfiguration approach, which allows new agents to join the formation, while maintaining the original configuration. Eventually, we give the relationships between our proposed method and the 2-D complex Laplacian method, which can be considered as a special case of our approach.

\vspace{12pt}


\begin{thebibliography}{99}

\bibitem{b1} Mu, Lan, et al. ``Motion behaviour based communication range estimation of adversarial drone swarms." IEEE Transactions on Network Science and Engineering, (2025).
\bibitem{b2} Xu, Jian, Shunxing Wei, and Liangang Yin. ``Novel distributed energy-saving control method for multi-AUV formation under long time-varying delays and clock errors." Ocean Engineering, 317 (2025): 119993.
\bibitem{b3} Dabass, Vandana, and Suman Sangwan. ``Strategic allocation: exploring optimization techniques in multi-robot systems." International Journal of Intelligent Robotics and Applications, (2025): 1-23.
\bibitem{b4} Amirkhani, Abdollah, and Amir Hossein Barshooi. ``Consensus in multi-agent systems: a review." Artificial Intelligence Review, 55.5 (2022): 3897-3935.
\bibitem{b5} Gulzar, Muhammad Majid, et al. ``Multi-agent cooperative control consensus: A comparative review." Electronics, 7.2 (2018): 22.
\bibitem{b6} Liu, Yefeng, et al. ``A survey of multi-agent systems on distributed formation control." Unmanned Systems, 12.05 (2024): 913-926.
\bibitem{b7}Ren, Wei. ``Multi-vehicle consensus with a time-varying reference state." Systems \& Control Letters, 56.7-8 (2007): 474-483.
\bibitem{b8}Ren, Wei, and Randal W. Beard. ``Consensus algorithms for double-integrator dynamics." Distributed Consensus in Multi-vehicle Cooperative Control: Theory and Applications, (2008): 77-104.
\bibitem{b9} De Marina, Hector Garcia, Bayu Jayawardhana, and Ming Cao. ``Distributed rotational and translational maneuvering of rigid formations and their applications." IEEE Transactions on Robotics, 32.3 (2016): 684-697.
\bibitem{b10}Zhao, Shiyu, and Daniel Zelazo. ``Bearing rigidity and almost global bearing-only formation stabilization." IEEE Transactions on Automatic Control, 61.5 (2015): 1255-1268.

\bibitem{b11}Zhao, Shiyu, and Daniel Zelazo. ``Translational and scaling formation maneuver control via a bearing-based approach." IEEE Transactions on Control of Network Systems, 4.3 (2015): 429-438.
\newpage
\bibitem{b12} Huang, Yunchang and Dai, Shi-Lu, ``Similarity-based rigidity formation maneuver control of underactuated surface vehicles over directed graphs." IEEE Transactions on Control of Network Systems, 12.1 (2025): 461-473.

\bibitem{b13}Lin, Zhiyun, et al. ``Distributed formation control of multi-agent systems using complex Laplacian." IEEE Transactions on Automatic Control, 59.7 (2014): 1765-1777.

\bibitem{b14}Fang, Xu, Lihua Xie, and Xiaolei Li. ``Distributed localization in dynamic networks via complex laplacian." Automatica, 151 (2023): 110915.
\bibitem{b15}de Marina, Hector Garcia. ``Distributed formation maneuver control by manipulating the complex Laplacian." Automatica, 132 (2021): 109813.
\bibitem{b16}Han, Zhimin, et al. ``Formation control with size scaling via a complex Laplacian-based approach." IEEE transactions on cybernetics, 46.10 (2015): 2348-2359.
\bibitem{b17}Fang, Xu, and Lihua Xie. ``Distributed formation maneuver control using complex laplacian." IEEE Transactions on Automatic Control, 69.3 (2023): 1850-1857.
\bibitem{b18}Zhao, Shiyu. ``Affine formation maneuver control of multiagent systems." IEEE Transactions on Automatic Control, 63.12 (2018): 4140-4155.
\bibitem{b19}Chen, Liangming, et al. ``Distributed leader–follower affine formation maneuver control for high-order multiagent systems." IEEE Transactions on Automatic Control, 65.11 (2020): 4941-4948.
\bibitem{b20}Xu, Yang, et al. ``Affine formation maneuver control of high-order multi-agent systems over directed networks." Automatica, 118 (2020): 109004.
\bibitem{b21}Zhu, Cheng, et al. ``Distributed affine formation maneuver control of autonomous surface vehicles with event-triggered data transmission mechanism." IEEE Transactions on Control Systems Technology, 31.3 (2022): 1006-1017.
\bibitem{b22}Xu, Yang, DeLin Luo, and HaiBin Duan. ``Distributed planar formation maneuvering of leader-follower networked systems via a barycentric coordinate-based approach." Science China Technological Sciences, 64.8 (2021): 1705-1718.
\bibitem{b23}Fang, Xu, Xiaolei Li, and Lihua Xie. ``Distributed formation maneuver control of multiagent systems over directed graphs." IEEE Transactions on Cybernetics, 52.8 (2021): 8201-8212.


\end{thebibliography}
\end{document}